\documentclass[runningheads,a4paper]{llncs}

\usepackage{epsfig,amsbsy}
\usepackage{amssymb,amsfonts}
\usepackage{graphicx}\graphicspath{{../../thesis_new/fig/pdf/}{../../thesis_new/fig/eps/}{../../thesis_new/fig/1_3/}{../../thesis_new/fig/3_1/}{../../thesis_new/fig/4_1/}{../../thesis_new/fig/5_1/}{../../thesis_new/fig/5_2/}{../../thesis_new/fig/6_1/}{../../thesis_new/fig/capsm/}{../../thesis_new/fig/literature/}{../../thesis_new/fig/mwbm/}}
\usepackage{mathtools}

\usepackage{multirow} 

\usepackage{color}
\usepackage[rightcaption]{sidecap}
\usepackage{caption}
\usepackage{subcaption}
\usepackage{wrapfig} 

\usepackage{cite}

\usepackage{algorithm,algorithmic}
\usepackage{graphics}
\usepackage{enumerate}
\usepackage{bm}      

\usepackage{soul}	
\usepackage{url}
\usepackage[pdfpagemode={UseOutlines},bookmarks=true,bookmarksopen=true,
   bookmarksopenlevel=0,bookmarksnumbered=true,hypertexnames=false,
   colorlinks=true,linkcolor={blue},citecolor={red},urlcolor={magenta},
   pdfstartview={FitV},unicode,breaklinks=true, pdftex]{hyperref} 

\usepackage{etoolbox} 
\AtEndEnvironment{example}{\null\hfill\qed}

\newtheorem{experiment}{Experiment}
\AtEndEnvironment{experiment}{\null\hfill\qed}%

\AtEndEnvironment{remark}{\null\hfill\qed}
\AtEndEnvironment{proof}{\null\hfill\qed}

\newcommand{\nint}{\mathbb{N}_0}

\newcommand{\cln}{\text{:\;}}

\usepackage[T1]{fontenc}
\usepackage[utf8]{inputenc}

\begin{document}

\mainmatter

\title{Experimental Evaluation of\\Modified Decomposition Algorithm for\\Maximum Weight Bipartite Matching\thanks{A preliminary version of this paper has been presented in the $11^{\textit{th}}$ International Conference on Theory and Applications of Models of Computation (TAMC 2014)~\cite{das14}. The current expanded version includes a better bound of the parameter $W'$ and the experimental evaluation of the theoretical claims made in previous version.}}

\titlerunning{Experimental Evaluation of Modified Decomposition Algo.\ for MWBM}

\author{Shibsankar Das}

\institute{Department of Mathematics\\
Institute of Science, Banaras Hindu University, 
Varanasi - 221\ 005, India.\\
\email{reach\_shibsankardas@yahoo.com, shibsankar@bhu.ac.in}
}

\date{\today}

\maketitle              

\begin{abstract}
Let $G$ be an undirected bipartite graph with 
positive integer weights
on the edges. We refine the existing decomposition
theorem originally proposed by Kao et al., for computing
maximum weight bipartite matching. We apply it to design an efficient version of the decomposition algorithm to compute the weight of
a maximum weight bipartite matching of $G$ in $O(\sqrt{|V|}W'/k(|V|,W'/N))$-time
by employing an algorithm designed by Feder and Motwani as a subroutine, 
where $|V|$ and $N$  denote the number of nodes and
the maximum edge weight of $G$, respectively and $k(x,y)=\log x /\log(x^2/y)$.
The parameter $W'$ is smaller than the total edge weight $W,$ 
essentially when the largest edge weight differs by more than one from the second largest edge weight in the current working graph in any decomposition step of the algorithm. In best case $W'=O(|E|)$ where $|E|$ be the number of edges of $G$ and 
in worst case $W'=W,$ that is, $|E| \leq W' \leq W.$
%
In addition, we talk about a scaling property of the algorithm and research a better bound of the parameter $W'$. 
An experimental evaluation on randomly generated data shows that the proposed improvement is significant in general.

\end{abstract}

\section{Introduction}
Let $G=(V = V_1 \cup V_2, E, \textit{Wt})$
be an undirected, weighted bipartite graph 
where $V_1$ and $V_2$ are two non-empty partitions of the vertex set $V$ of $G$, and $E$ is the edge set of $G$ with
positive integer weights on the edges which are given by the weight function 
$\textit{Wt} \cln E \rightarrow \mathbb{N} $, where $\mathbb{N}$ is the set of positive integers.  
%
%
%
Throughout the paper, we use the symbols $N$ and  $W$  to denote the largest weight of any edge and  the total weight of $G$, respectively.  
The weight of the graph $G$  is defined by
$W=\textit{Wt}(G)=  \sum_{e \in E} \textit{Wt}(e)$. 
We also assume that the graph does not have any isolated vertex. 
For uniformity we treat an unweighted
graph as a weighted graph having unit weight for all edges.

We use the notation $\{u,v\}$ for an edge $e \in E$ between $u \in V_1$ and $v \in V_2$, and its weight is denoted by $\textit{Wt}(e)=\textit{Wt}(u,v)$.
We also say that $e=\{u,v\}$ is \textit{incident} on vertices $u$ and $v$;
and $u$ and $v$ are each \textit{incident} with $e$.
Two vertices $u,v \in V$ of $G$ are \textit{adjacent} if there exists an edge $e=\{u,v\}\in E$ of $G$ to which they are both incident. 
Two edges $e_1,e_2 \in E$ of $G$  are  \textit{adjacent} if there exists a vertex $v\in V$ to which they are both incident. 

A subset $M \subseteq E$ of edges is a \emph{matching} if no two edges of $M$
share a common vertex. A vertex $v \in V$ is said to be \emph{covered} or
\emph{matched} by the matching $M$ if it is incident with an edge of
$M$; otherwise $v$ is \emph{unmatched}~\cite{bondy82,bondy08}.

A matching $M$ of $G$ is called a \textit{maximum} (\textit{cardinality}) \textit{matching} if there
does not exist any other matching of $G$ with greater cardinality. We denote such a
matching by $\textit{mm}(G)$. The weight of a matching $M$ is defined as
$\textit{Wt}(M) = \sum_{e \in M} \textit{Wt}(e)$. A matching $M$ of $G$ is a \emph{maximum weight
matching}, denoted as $mwm(G)$, if $\textit{Wt}(M) \geq \textit{Wt}(M')$ for every other matching
$M'$ of the graph $G$.

Observe that, if $G$ is an unweighted graph then $\textit{mwm}(G)$ is a $\textit{mm}(G)$, which we write as $\textit{mwm}(G)=\textit{mm}(G)$ in short and its
weight is given by $\textit{Wt}(\textit{mwm}(G))$ $=|\textit{mm}(G)|$. Similarly, if $G$ is an
undirected and weighted graph with $\textit{Wt}(e) = c$ for all edges $e$ in $G$
and $c$ is a constant then also we have $\textit{mwm}(G)=\textit{mm}(G)$ with weight of the
matching as $\textit{Wt}(\textit{mwm}(G))=c*|\textit{mm}(G)|$.
%
%
%
\subsection{Our Contribution}
In~\cite{kao99,kao02}, Kao et al.\ proposed a decomposition theorem and algorithm for computing weight of a Maximum Weight Bipartite Matching (MWBM) of the bipartite graph $G$.
Our contribution in this paper is a revised version of the existing decomposition theorem
and use it efficiently to design an improved version of the decomposition algorithm to estimate the weight of a MWBM of $G$ in time
$O(\sqrt{|V|}W'/k(|V|,W'/{N}))$ by taking algorithm designed by Feder and
Motwani~\cite{feder95} as base algorithm, where
$k(x,y)=\log x /\log(x^2/y)$. 

This algorithm bridges a gap between the best known time complexity of computing a Maximum Cardinality Matching (MCM) and that of computing a MWBM of a bipartite graph.
In best case, computation of weight of a MWBM takes 
$O(\sqrt{|V|}|E|/k(|V|,|E|))$ time which is the same as the complexity of the Feder and Motwani's algorithm~\cite{feder95} for computing MCM of unweighted bipartite graph;  whereas
in worst case it takes $O(\sqrt{|V|}W/k(|V|,W/{N}))$, i.e.,\ $|E| \leq W' \leq W$. 
Further, we provide an interesting scaling property of the algorithm and a better bound of the parameter $W'$.
%
%
However, it 
seems to be a challenging problem to get rid of $W$ or $N$ from the complexity.

The modified algorithm works well for general $W,$ but is best known for
$W'=o(|E| \log(|V| N))$. We also design a revised algorithm to construct
minimum weight cover of a bipartite graph in time  
$O(\sqrt{|V|}W'/k(|V|,W'/{N}))$
to identify the edges involved in maximum weight bipartite matching.
It is also possible to use other algorithms as a subroutine,
for example, algorithms given by Hopcroft and Karp \cite{hopcroft73} and
Alt et al.~\cite{alt91} in which case the running times of our algorithm will be $O(\sqrt{|V|}W')$ and 
$O((|V|/ \log |V|)^{1/2}W')$, respectively.
An experimental evaluation on randomly generated bipartite
graphs shows that the proposed improvement is significant in general.
%
\subsection{Roadmap}
In Section~\ref{mwbm:Survey},
we give a detailed summary of existing maximum matching algorithms and
their complexities for unweighted and weighted bipartite graphs.
Section~\ref{Decomposition} describes modified decomposition theorem and an
algorithm to compute the weight of a MWBM. The complexity analysis of the
algorithm is discussed in Section~\ref{mwbm:Complexity_analysis}. The algorithm to 
compute minimum weight cover of a bipartite graph is given in Section~\ref{Find_MWC}, which is used to find the edges of a MWBM. 
Section~\ref{mwbm:Experiments} provides the experimental comparisons between the modified algorithm and Kao et al.'s algorithm for randomly generated bipartite graphs.
We summarize the results in Section~\ref{mwbm:Conclusion}.

%
%
\section[Survey of Maximum Matching in Bipartite Graph]{Survey of Maximum Matching in Bipartite Graph}
\label{mwbm:Survey}
The problem of computing maximum matching in a given graph is one of
the fundamental algorithmic problem that has played an important role
in the development of combinatorial optimization and algorithmics.
A survey of some of the well known existing maximum (cardinality) matching 
and maximum weight  matching algorithms for bipartite graph are summarized in
Table \ref{Table:MUWBM}  and Table \ref{Table:MWBM}, respectively.
The algorithms with best asymptotic bound are indicated by ``$*$''
in these tables. A more detailed and technical discussion of the algorithms can be found
in textbooks~\cite{korte07,schrijver03,douglas00}.

\subsection{Maximum Cardinality Matching}
For unweighted bipartite graphs, Hopcroft-Karp \cite{hopcroft73}
algorithm, which is based on augmenting path technique, offers
the best known performance for finding maximum matching in time
$O(|E|\sqrt{|V|})$. In case of dense unweighted bipartite graphs, that is with
$|E|=\Theta (|V|^2)$, slightly better algorithms exist.
An algorithm by Alt et al.~\cite{alt91} obtains a maximum matching in
$O(|V|^{1.5}\sqrt{ |E|/\log |V|})$ time. In case of $|E|=\Theta (|V|^2)$, this
becomes $O(|E|{\sqrt{|V|/\log |V|}})$ and is also 
${\sqrt{\log |V|}}$-factor
faster than Hopcroft-Karp algorithm. This speed up is obtained by an
application of the fast adjacency matrix scanning technique of Cheriyan,
Hagerup and Mehlhorn~\cite{cheriyan90}. The algorithm proposed by Feder-Motwani~\cite{feder95} has the time complexity 
$O(|E|\sqrt{|V|}/k(|V|,|E|))$, where
$k(x,y)={\log x} /\log(x^2/y)$.
\vfill
\begin{table*}[htpb]
\centering
\caption{Complexity survey of maximum unweighted bipartite matching algorithms.}
\label{Table:MUWBM}
\begin{tabular}{|l|l|l|} 	
\hline
\multicolumn{1}{|c|}{\bf Year}	& \multicolumn{1}{|c|}{\bf Author(s)}	& \multicolumn{1}{|c|}{\bf Complexity}\\
\hline
1973 $*$ & Hopcroft and Karp \cite{hopcroft73}  		& $O(|E|\sqrt{|V|})$	\\
\hline
1991 & Alt, Blum, Mehlhorn and Paul \cite{alt91} 	& $O(|V|^{1.5}\sqrt{ |E|/\log |V|})$ \\
\hline
1995 $*$ &Feder and Motwani \cite{feder95}	  		& $O(|E|\sqrt{|V|}/k(|V|,|E|))$ \\
\hline
\end{tabular}
\end{table*}
%


\vspace*{1cm}
\subsection{Maximum Weight Bipartite Matching}
Several algorithms have also been proposed for computing maximum weight
bipartite matching, improving both theoretical and practical running
times. The well known Hungarian method, the first polynomial time
algorithm, was introduced by Kuhn \cite{kuhn55} and Munkres
\cite{munkres57}. Fredman and Tarjan \cite{fredman87} improved this with
running time $O(|V|(|E| + |V| \log |V|))$ for sparse graph by using Fibonacci heaps. An 
$O(|V|^{3/4} |E| \log N)$-time scaling algorithm was proposed by
Gabow \cite{gabow85} under the assumption
that edge weights are integers. A different and faster scaling algorithm was given by
Gabow and Tarjan~\cite{gabow89} with running time
$O(\sqrt{|V|} |E|\log(|V|N))$. Kao et al.~\cite{kao02} proposed an
$O(\sqrt{|V|}W/k(|V|,W/N))$-time decomposition technique
under the assumptions that weights on the edges are positive and $W=o(|E|\log(|V|N))$.

In addition to the above exact algorithms, several randomized and approximate algorithms
are also proposed, see for example~\cite{duan10,sankowski09}.
For a tight lower bound for the weights of maximum weight matching in bipartite graph, please refer to~\cite{das16_TLB_arXiv}.

\begin{table*}[htpb]
\centering
\caption{Complexity survey of maximum weight bipartite matching algorithms.}
\label{Table:MWBM}
\begin{tabular}{|l|l|l|} 	
\hline
 \multicolumn{1}{|c|}{\bf Year(s)}		&  \multicolumn{1}{|c|}{\bf Author(s)}	& \multicolumn{1}{|c|}{\bf Complexity} \\
\hline
1955,				& Kuhn \cite{kuhn55},				& $O(|V|^4)$ 		\\
1957				& Munkres \cite{munkres57} 				& (Hungarian method)	\\
\hline
1960				& Iri \cite{iri60,schrijver03} 			& $O(|V|^2 |E|)$	\\
\hline
1969				& Dinic and Kronrod \cite{dinic69,schrijver03} 			& $O(|V|^3)$	\\
\hline
1984, 1987 $*$	&Fredman and Tarjan \cite{fredman87}  	& $O(|V|(|E|+|V|\log |V|))$	\\
\hline
1985  			&Gabow \cite{gabow85} 					& $O(|V|^{3/4} |E| \log N)$	\\
\hline
1989 $*$			& Gabow and Tarjan \cite{gabow89}	 	& $O(\sqrt{|V|} |E| \log(|V|N))$	\\
\hline
1999 			& Kao, Lam, Sung and Ting \cite{kao99}		 		& $O(\sqrt{|V|}W)$  \\
\hline
2001	 $*$			& Kao, Lam, Sung and Ting \cite{kao02} 				& $O(\sqrt{|V|}W/k(|V|,W/N))$  \\
\hline
2014 $*$ & 		 & $O(\sqrt{|V|}W')$\\
\cline{3-3}
(This work)	& 	& $O((|V|/\log |V|)^{1/2}W')$	\\
\cline{3-3}
			& 	& $O(\sqrt{|V|}W'/k(|V|,W'/{N}))$	\\
\cline{3-3}
\hline
\end{tabular}
\end{table*}

\section[Refined Decomposition Theorem for MWBM]
{Refined Decomposition Theorem for Maximum Weight Bipartite Matching}
\label{Decomposition}
We now propose a modified decomposition theorem which is a generalization of the existing decomposition theorem originally proposed by Kao et al.~\cite{kao02,kao99}
and use it to develop a revised  version of the decomposition algorithm to decrease the
number of iterations and speed up the computation of the weight of a MWBM.
Let $G = (V = V_1 \cup V_2 , E, \textit{Wt})$ be an undirected, weighted bipartite graph 1
having $V_1$ and $V_2$ as partition
of vertex set $V$. Further, let $E=\{e_1,e_2,\ldots,e_{|E|}\}$ be set of
edges of $G$ with weights $\textit{Wt}(e_i)=w_i$ for $1 \leq i \leq |E|$, where  $w_1, w_2, \ldots, w_{|E|}$ are
not necessarily distinct. As defined earlier, let $N$ be the maximum edge
weight, that is, for all $i \in \{1,2,\ldots, |E|\}$, $0 \leq w_i \leq N$, and
 $W=\sum_{1 \leq i \leq |E|} w_i $ be the total weight of  $G$.

Our algorithm considers several intermediate graphs with  lighter edge weights. During this process it is possible that weights of some of the
edges may 
become zero. 
An edge $e \in E$ is said to be {\it active} if
its weight $\textit{Wt}(e)>0$, otherwise it is said to be {\it inactive}, that is when $\textit{Wt}(e)=0$. Let there be $m' ~(\leq |E|)$ distinct edge weights in current
working graph where $w_1 < w_2 < \cdots < w_{m' -1} < w_{m'}$.
We denote the first two distinct maximum edge weights in current
working graph by $H_1$ and $H_2 ~(< H_1)$, respectively.
Assign $H_2 = 0$ in case $m' = 1$.

We first build two new graphs referred to as $G_h$ and $G_h^\Delta$ from
a given weighted bipartite graph $G$. For any integer $h \in [1,N]$, we decompose the 
graph $G$ into two lighter weighted bipartite graph $G_h$ and
$G_h^\Delta$ as proposed by Kao et al.~\cite{kao99,kao02}. 
A minimum weight cover is a dual of maximum weight matching \cite{kao02}. A \emph{cover}
of $G$ is a function $C \cln V_1 \cup V_2 \rightarrow \nint$ such that
$C(v_1)+C(v_2) \geq \textit{Wt}(v_1,v_2) ~~ \forall v_1 \in V_1 \mbox{ and } v_2 \in V_2$. Let $\textit{Wt}(C)= \sum_{x\in V_1 \cup V_2} C(x)$. A cover $C$ is \emph{minimum weight cover} if $\textit{Wt}(C)$ is minimum.
\begin{description}

\item[Formation of $\boldsymbol{G_h}$ from $\boldsymbol G$:] The graph
$G_h$ is formed by including those edges $\{u,v\}$ of $G$ whose weights
 $\textit{Wt}(u,v)$ lie in the range $[N-h+1, N]$. Each edge $\{u,v\}$ in graph
$G_h$ is assigned weight $\textit{Wt}(u,v)-(N-h)$. For illustration, $G_1$ is
constructed by the maximum weight edges of $G$ and assigned unit weight
to each edge.

\item[Formation of $\boldsymbol{G_h^\Delta}$ from $\boldsymbol G$:] Let
$C_h$ be the minimum weight cover of $G_h$. The graph $G_h^\Delta$ is
formed by including every edge $\{u,v\}$ of $G$ whose weight satisfies the condition 
$$\textit{Wt}(u,v)-C_h(u)-C_h(v) > 0.$$ 
The weight assigned to such an
edge is $\textit{Wt}(u,v)-C_h(u)-C_h(v)$.
\end{description}

\begin{theorem}[The Decomposition Theorem \cite{kao02}]
\label{DecompositionTh}
Let $G$ be an undirected, weighted bipartite graph.
Then
\begin{enumerate}[(a)]
\item for any integer $h\in[1,N]$, $$\textit{Wt}(\textit{mwm}(G))= \textit{Wt}(\textit{mwm}(G_h))+
\textit{Wt}(\textit{mwm}(G_h^\Delta)),$$
\item in particular (trivial), for $h=1$,
$$\textit{Wt}(\textit{mwm}(G))= \textit{Wt}(\textit{mm}(G_1))+ \textit{Wt}(\textit{mwm}(G_1^\Delta)).$$
\end{enumerate}
\end{theorem}

Note that the Theorem \ref{DecompositionTh}(b) is derived from
Theorem \ref{DecompositionTh}(a), since for $h=1$, we have $$\textit{mwm}(G_1)=
\textit{mm}(G_1)$$ and 
$$\textit{Wt}(\textit{mwm}(G_1))=\textit{Wt}(\textit{mm}(G_1))=|\textit{mm}(G_1)|.$$ 
The Theorem
\ref{DecompositionTh}(b) is used recursively in the Algorithm~\ref{Algorithm0_Kao}~\cite{kao02},
to compute the weight of a maximum weight matching of the graph $G$.


\begin{algorithm}[H]
{
\caption[Kao et al.'s algorithm to compute weight of a MWBM.]{Kao et al.'s algorithm~\cite{kao02} to compute weight of a MWBM.}
\label{Algorithm0_Kao}
\begin{description}
\item[\it Input:] A weighted, undirected, complete bipartite graph $G$ with
positive integer weights on the edges.
\item[\it Output:] Weight of a maximum weight matching of $G$, that is, $\textit{Wt}(\textit{mwm}(G))$.
\end{description}
\begin{tabbing}
123\=\kill
{Compute-$\textsc{Mwm}(G)$}\+\\
1: \= Construct $G_1$ from $G$.\\
2:   \> Compute $\textit{mm}(G_1)$ and find a minimum weight cover $C_1$ of $G_1$.\\
3:   \> Construct $G_1^\Delta$ from $G$ and $C_1$.\\
4:   \> \textbf{if} $G_1^\Delta$ is empty,\\
5:\>~~~ \textbf{then} \textbf{return} $\textit{Wt}(\textit{mm}(G_1))$;\\
6:\> \textbf{else} \textbf{return} $\textit{Wt}(\textit{mm}(G_1))$+Compute-$\textsc{Mwm}(G_1^\Delta)$.
\end{tabbing}
}
\end{algorithm}

\begin{remark}
A graph $G$ may not have all edge weights distinct. Consider the set of distinct edge
weights of $G$. The Algorithm~\ref{Algorithm0_Kao} works efficiently only when the largest
edge weight differs by exactly one from the second largest edge weight of the current graph during an invocation of Theorem~\ref{DecompositionTh}(b) in each iteration.  
\end{remark}


\begin{remark}
Observe that for arbitrary $h \in [1,N]$, $\textit{mwm}(G_h)$ need not be equal
to $\textit{mm}(G_h)$, that is, we cannot always conclude that $\textit{mwm}(G_h)=\textit{mm}(G_h)$.
\end{remark}

One of our objectives is to investigate those values of $h$ for which
$\textit{mwm}(G_h)$ is equal to $\textit{mm}(G_h)$ apart from
the trivial value of $h$
as 1 in each iteration of the Algorithm~\ref{Algorithm0_Kao} to generate $G_h$ having all its
edge weights as 1.

In order to get the speed up whenever possible, by decreasing the number of iterations whenever possible, we revise
the Theorem~\ref{DecompositionTh}(b) and propose Theorem~\ref{DecompositionTh_Modified} which gives a
domain of $h \in [1,N]$ where $\textit{mwm}(G_h)=\textit{mm}(G_h)$ and as a consequence of that 
we can write 
$$\textit{Wt}(\textit{mwm}(G_h))=\textit{Wt}(\textit{mm}(G_h))=h*|\textit{mm}(G_h)|.$$ 
It works for $h=1$
and performs well especially when the largest edge weight differs  by
more than one from the second largest edge weight in the current graph
in a decomposition step during an iteration.

\begin{theorem}[The Modified Decomposition Theorem]
\label{DecompositionTh_Modified}
The following equalities hold for any integer $h \in [1,H_1-H_2]$ where
$H_1$ and $H_2 ~(<H_1)$ are the first two distinct maximum edge weights of graph
$G$, respectively. We assign $H_2 = 0$ in case all edge weights are
equal.
\begin{enumerate}[(a)]
\item \label{thm:dt:a} $\textit{mwm}(G_h)  = \textit{mm}(G_h)$,

\item
$\textit{Wt}(\textit{mwm}(G)) = h * \textit{Wt}(\textit{mm}(G_h))+ \textit{Wt}(\textit{mwm}(G_h^\Delta)).
$
\end{enumerate}
\end{theorem}

\begin{proof}
The proof of the above statements are based on the construction of new
graphs $G_h$ and $G_h^\Delta$ from $G$ and Theorem \ref{DecompositionTh}(a).
\begin{enumerate}[(a)]
\item To prove that for any integer $h$ where $1 \leq h \leq H_1-H_2$,
$\textit{mwm}(G_h)=\textit{mm}(G_h)$ holds true, it is enough to prove the
same for the maximum value\footnote{
For illustration,
consider $h=c$ where $1 \leq c \leq H_1-H_2$. Then as per the
formation of $G_h$ from $G$, $G_c$ is built by choosing those edges
of $G$ that have weight $\textit{Wt}(u,v) \in [N-(c-1),N]$. Since, $c-1 \geq 0$
and $N \in [N-(c-1),N]$ for any $c \in [1, H_1-H_2]$, $G_c$ has only
the heaviest edges of $G$. For optimization, choose $h=H_1-H_2$,
the maximum possible value of $h$.} of
$h$, that is, for $h=H_1-H_2$. As specified earlier, the
construction of $G_h$ is done by choosing those edges $\{u,v\}$ of $G$ that have weight 
$$\textit{Wt}(u,v) \in [N-h+1,N]=[H_1-(H_1-H_2)+1, H_1]=[H_2+1,H_1].$$
Since $H_1 \in [H_2+1, H_1]$, $G_h$ has only the heaviest edges of
$G$ and each such edge is assigned the same weight.
Thus, $\textit{mwm}(G_h)=\textit{mm}(G_h)$ for $h=H_1-H_2$.

\item Observe that $h \in[1,H_1-H_2]$ and $[1,H_1-H_2]\subseteq [1,N]$.
So, by using Theorem \ref{DecompositionTh}(a) we have, $\forall h \in[1,H_1-H_2]$,
$$\textit{Wt}(\textit{mwm}(G))= \textit{Wt}(\textit{mwm}(G_h))+\textit{Wt}(\textit{mwm}(G_h^\Delta)).$$
Also by using the Theorem~\ref{DecompositionTh_Modified}(\ref{thm:dt:a}),
$\textit{mwm}(G_h)=\textit{mm}(G_h)$ for all $h \in [1,H_1-H_2]$.
Weight of each edge\footnote{Only maximum weight edges of $G$ are 
included in $G_h$.} 
$\{u,v\}$ in $G_h$ is exactly $\textit{Wt}(u,v)-(N-h) =
H_1-(H_1-h) = h$. Therefore,
\[\textit{Wt}(\textit{mwm}(G_h))=h * \textit{Wt}(\textit{mm}(G_h))= h * |\textit{mm}(G_h)|.\]
Hence for any integer $h \in [1,H_1-H_2]$, 
\begin{align*}
  \textit{Wt}(\textit{mwm}(G)) &= \textit{Wt}(\textit{mwm}(G_h))+ \textit{Wt}(\textit{mwm}(G_h^\Delta))\\[2pt]
			 & = h * \textit{Wt}(\textit{mm}(G_h))+ \textit{Wt}(\textit{mwm}(G_h^\Delta)).
\end{align*}
\end{enumerate}
This completes the proof.
\end{proof}

\begin{remark}
The equality $\textit{mwm}(G_h)=\textit{mm}(G_h)$ in Theorem
\ref{DecompositionTh_Modified}(\ref{thm:dt:a}) is not true for
$h > H_1-H_2$ and $h \leq N$.
\end{remark}

To show that for any $h \in [H_1-H_2 + 1, N]$ the statement
$\textit{mwm}(G_h)=\textit{mm}(G_h)$ is not true, it is enough to show the same
essentially for $h=H_1-H_2+1$. Observe that $h = H_1-H_2+1 \geq 2$,
since $H_1 > H_2$. According to the construction of $G_h$, it is formed
by edges $\{u,v\}$ of $G$ whose weights $\textit{Wt}(u,v) \in [N-h+1,N] = [H_1-(H_1-H_2+1)+1,
H_1]=[H_2,H_1]$, that is, $G_h$ is built with the maximum weight edges and second
maximum weight edges of $G$,  because $\{H_1,H_2\}\in [H_2, H_1]$.
The weight of each 
heaviest edge $\{u,v\}$ of $G$ in $G_h$
is exactly 
$$\textit{Wt}(u,v)-(N-h) = H_1-(H_1-h) = h$$ which is greater than or equal to 2 and that of each second
heaviest edge $\{u,v\}$ of $G$ in $G_h$ is exactly $$\textit{Wt}(u,v)-(N-h) = H_2-(H_1-h) =
(H_2-H_1)+h=(1-h)+h=1.$$ 
Hence $\textit{mwm}(G_h) \neq \textit{mm}(G_h)$ for such
a value of $h$.

\begin{example}
Consider the graph shown in the Figure \ref{Diagram3}(a). Let $h =
H_1-H_2+1$. So, $h=H_1-H_2+1=9-4+1=6$.
As shown in the Figure \ref{Diagram3}(b), $G_h$ is
formed by the edges $\{u,v\}$ whose weights $\textit{Wt}(u,v) \in [N-h+1,N]=[9-6+1,9] =[4,9]$ and their respective calculated weights are 6 and 1. Hence
$\textit{mwm}(G_h) \neq \textit{mm}(G_h)$.
%
\begin{figure*}[htpb]
\centering
\includegraphics[angle=0,width=.8\textwidth]{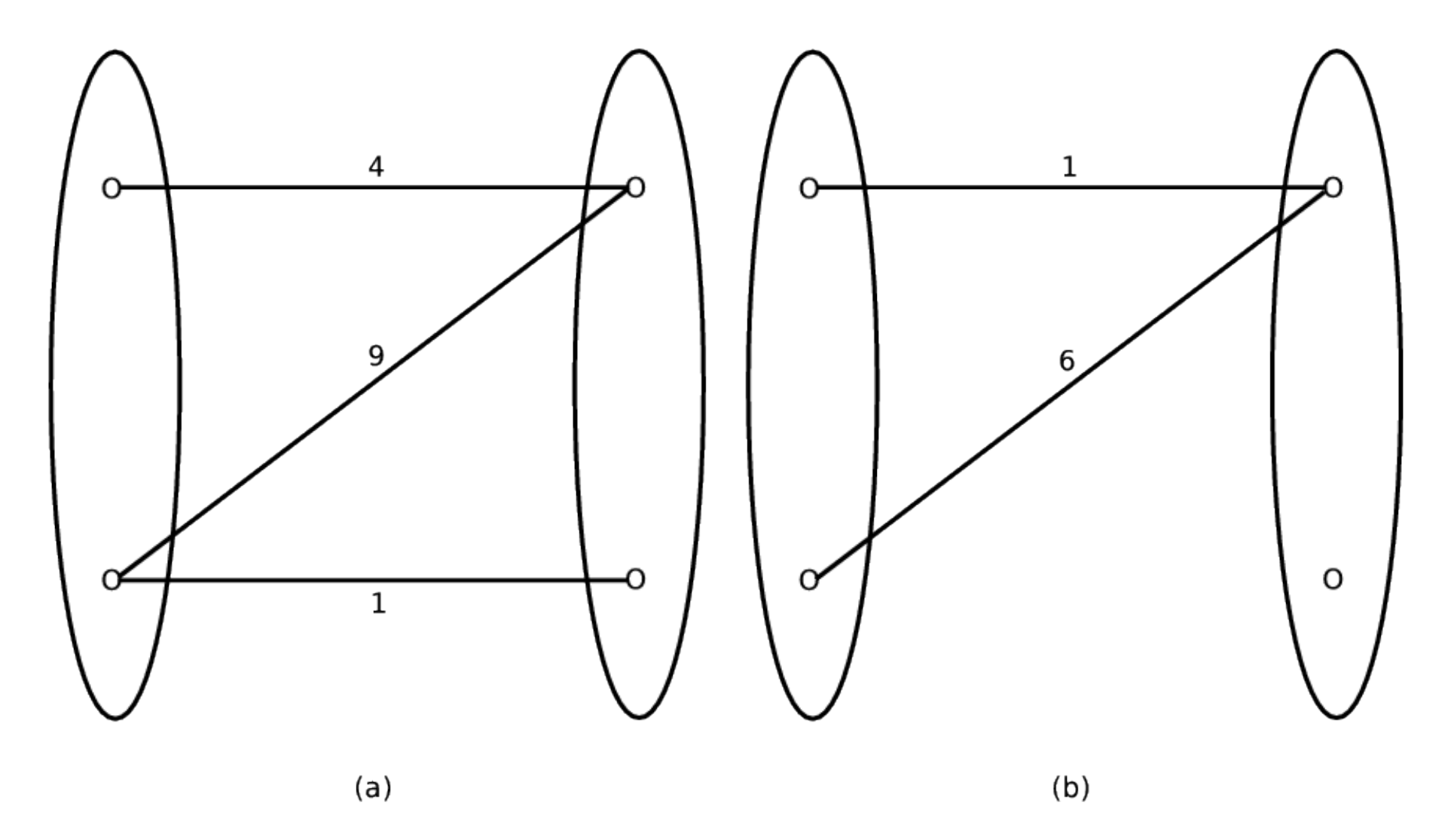}
\caption[Example for $H_1-H_2 < h \leq N$, $\textit{mwm}(G_h) \neq \textit{mm}(G_h)$.]{{\bf(a)} An undirected bipartite graph $G$ with 
positive integer weights on the edges.  {\bf(b)} Considering $h=H_1-H_2+1=6$,
$G_h$ is extracted, but $\textit{mwm}(G_h) \neq \textit{mm}(G_h)$.}
\label{Diagram3}
\end{figure*}
\end{example}

We use the modified decomposition Theorem~\ref{DecompositionTh_Modified} to design a recursive Algorithm~\ref{Algorithm1} to compute the weight of a $\textit{mwm}(G)$.

\begin{algorithm}[H]
\caption{Compute weight of a maximum weight matching of $G$.}
\label{Algorithm1}
\begin{description}
\item[\it Input:] A weighted, undirected, complete bipartite graph $G$ with 
positive  integer weights on the edges.
\item[\it Output:] Weight of a maximum weight matching of $G$, that is, $\textit{Wt}(\textit{mwm}(G))$.
\end{description}
\begin{tabbing}
123\=\kill
\textsc{Wt-Mwbm($G$)}\+\\
1: Assume that initially $\textit{Wt}(\textit{mwm}(G))=0$.\\
2: Find $h=H_1-H_2$ from the current working graph $G$.\\
3: Construct $G_h$ from $G$.\\
4: Compute $\textit{mm}(G_h)$.\\
5: Find minimum weight cover $C_h$ of $G_h$.\\
6: Construct $G_h^\Delta$ from $G$ and $C_h$.\\
7: \= {\bf if} $G_h^\Delta$ is empty (that is, $G_h^\Delta$ has no active edge)\\
8: \>~~~ {\bf then return}  $h * |\textit{mm}(G_h)|$;\\
9: \> {\bf else return} $h * |\textit{mm}(G_h)|$ + \textsc{Wt-Mwbm($G_h^\Delta$)}.
\end{tabbing}
\end{algorithm}

\begin{example}
Consider the bipartite graph shown in Figure \ref{Diagram1}(a). The Algorithm
\ref{Algorithm1} finds the weight of a MWBM in just two iterations, as
the algorithm is designed for the best $h$ in every invocation of \textsc{Wt-Mwbm(\,)},
whereas algorithm by Kao et al.~\cite{kao02} requires 500 iterations because it
considers $h=1$ in every invocation of  Compute-\textsc{Mwm}(\,).
%
\begin{figure*}[htpb]
\centering
\includegraphics[angle=0,width=1\textwidth]{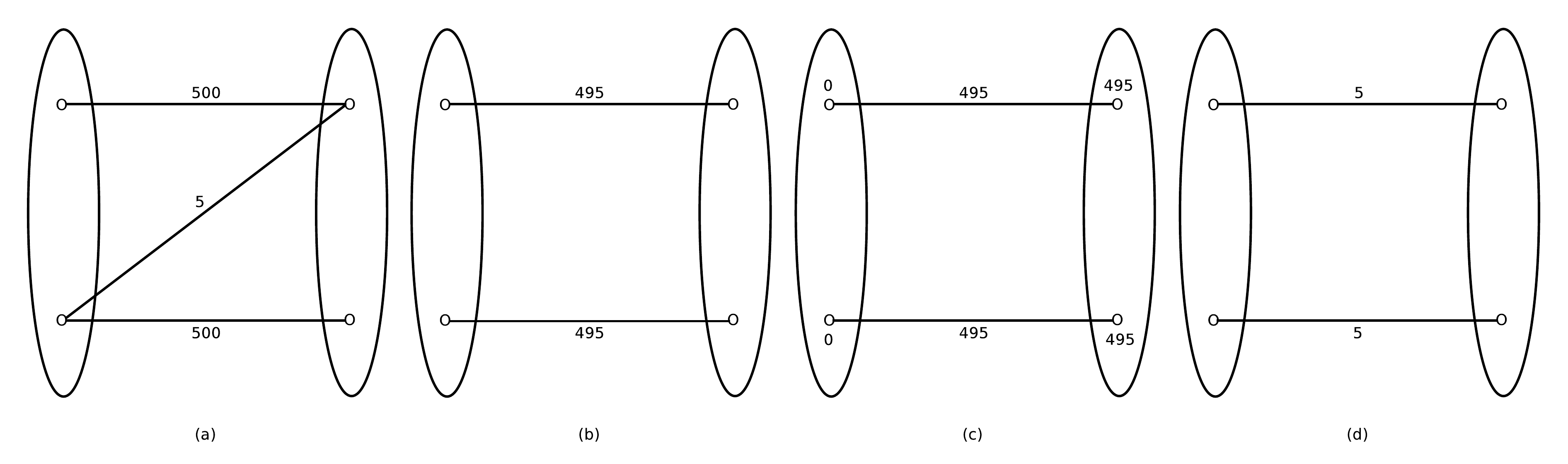}
\caption[Illustration of the modified decomposition theorem on  an undirected, weighted bipartite graph $G$.]{{\bf(a)} An undirected, weighted bipartite graph $G$ with 
positive integer weights on the edges. In the  current graph $G$, $h=495$. {\bf(b)}
$G_h$ is extracted. {\bf(c)} $C_h$ is the weighted cover of $G_h$.
{\bf(d)} $G_h^\Delta$ is formed from $G_h$ and $C_h$. Compute 
\textsc{Wt-Mwbm$(G_h^\Delta)$}.}
\label{Diagram1}
\end{figure*}
\end{example}

Correctness of the algorithm follows from the construction of $G_h$ and
${G_h^\Delta}$ and the modified decomposition Theorem
\ref{DecompositionTh_Modified}.


\section{Complexity of the Modified Algorithm}
\label{mwbm:Complexity_analysis}
Let $G = (V = V_1 \cup V_2 , E, \textit{Wt})$ be the initial input graph and $N$ denotes the maximum edge weight of $G$, that is, for all $i \in \{1,2,\ldots, |E|\}$, $0
\leq w_i \leq N$ and $W=\sum_{1 \leq i \leq |E|} w_i $ is the total weight of $G$. Further, let $\{w_1, \ldots, w_{m'}\}$ be the set
of distinct edge weights of $G$, where $m' \leq |E|$.

Based on the constructions of $G_h$ and $G_h^\Delta$,  the modified
decomposition Theorem~\ref{DecompositionTh_Modified} and the
Algorithm~\ref{Algorithm1},
we can easily observe that in worst case the maximum number of possible iterations of \textsc{Wt-Mwbm}(\,)
is $N$, when $h=1$ in each iteration in the current working graph. 
Whereas in the best case,
all the edge weights of $G$ are the same and so we will have $h=N$ for the present decomposition. As a consequence the algorithm will terminate in the
first iteration itself.

As the complexity analysis of the Algorithm~\ref{Algorithm1} is almost
similar to that presented elsewhere~\cite{kao02}, the details are available in Appendix~\ref{AppendixB: Detailed Complexity of MWBM} (see page~\pageref{AppendixB: Detailed Complexity of MWBM}).
The algorithm takes
$O(\sqrt{|V|}W'/k(|V|,W'/{N}))$ time to compute the weight of a
$\textit{mwm}(G)$ by using the algorithm by Feder and Motwani~\cite{feder95}, as a subroutine.

Let $L_i$ consists of edges of remaining $G$ (after \,$i-1$-th iteration) whose weights reduce in $G_h^\Delta$ in $i$-th iteration. 
Also let there be $p$ iterations, $l_i=|L_i|$ for
$i=1,2,\ldots,p  \leq N$ and $h_i=H_{i1}-H_{i2}$ in the $i$-th iteration, where $H_{i1}$ and $H_{i2} \;(<H_{i1})$ are the first two distinct maximum edge weights of the remaining graph $G$ after the \,$i -1$-th iteration.

From the detailed complexity analysis we have,  $l_1h_1+l_2h_2+ \cdots + l_ph_p=W.$ Let $l_1+l_2+ \cdots + l_p=W'.$ Observe that, in worst case, if
$h_i=1$ for all $i \in [1,p]$, then $W'=\sum_{i=1}^p l_i=W.$ And in best case, if $h_1=N$, then $W'=|E|$. 
Moreover, the parameter $W'$ is smaller than $W$,
 essentially when the largest edge weight differs by more than one from the
second largest edge weight in the current working graph in  decomposition step
during at least one iteration of the algorithm. Therefore  in best case\footnote{
In best case, all the edge weights of $G$ are the same. So, the algorithm terminates
in just one iteration and hence $W'=O(|E|).$}, it
requires $O(\sqrt{|V|}|E|/k(|V|,|E|))$ time 
and in worst case $O(\sqrt{|V|}W/k(|V|,W/{N})),$  to compute weight of a maximum weight matching. That is,  $|E| \leq W' \leq W$.

This time complexity bridges a gap between the best known time complexity for computing a 
Maximum Cardinality Matching (MCM) of unweighted bipartite graph and that of computing a MWBM of a weighted bipartite graph.
In best case, for computation of weight of a MWBM, the Algorithm~\ref{Algorithm1} takes 
$O(\sqrt{|V|}|E|/k(|V|,|E|))$ time which is the same as the complexity of the Feder and Motwani's algorithm~\cite{feder95} for computing MCM of unweighted bipartite graph;  whereas
in worst case it (Algorithm~\ref{Algorithm1}) takes $O(\sqrt{|V|}W/k(|V|,W/{N}))$ time which is the same as the complexity of the Kao et al.'s algorithm~\cite{kao02}.
However, it is very difficult and challenging to get rid of $W$ or $N$ from the complexity.
This modified algorithm works well for general $W,$ but is best known for
$W'=o(|E| \log(|V|N))$.
%

\subsection{More Advantages: Scaling Up and Down, and GCD Properties}

Some other advantages of the modified decomposition algorithm 
is stated by the following propositions. 
Let $G = (V, E, \textit{Wt})$ be an undirected, weighted bipartite graph, $E=\{e_1,e_2,\ldots,e_{|E|}\}$ be the set of positive integer weight
edges with weights $\textit{Wt}(e_i)=w_i >0$ (where $1 \leq i \leq |E|$), $N$ be the maximum edge weight and
 $W=\sum_{1 \leq i \leq |E|} w_i $ be the total weight of  $G$.
The modified decomposition Algorithm~\ref{Algorithm1} computes weight of a maximum weight bipartite  matching  of $G$ in 
 $O(\sqrt{|V|}W'/k(|V|,W'/{N}))$ time, where $|E| \leq W' \leq W$.
 %
\begin{proposition}[Multiplicative Scaling Up Property]
\label{Lm:M_Incremental Wt}
Let $\widehat{G}$  be a new weighted bipartite graph constructed by multiplying a large constant $\alpha \in \mathbb{N} $ to each edge weight $w_i$ 
of the initial weighted bipartite graph $G$. 
Then  for both the graphs $G$ and $\widehat{G}$, the complexity of the Algorithm~\ref{Algorithm1} remains  $O(\sqrt{|V|}W'/k(|V|,W'/{N}))$ where $|E| \leq W' \leq W$; whereas for the graph $\widehat{G}$, the complexity of the Algorithm~\ref{Algorithm0_Kao} becomes $O(\sqrt{|V|}\alpha W/k(|V|,\alpha W/{N}))$.
\end{proposition}
\begin{proof}
As mentioned in the  detailed complexity analysis of the Algorithm~\ref{Algorithm1} (described in Appendix~\ref{AppendixB: Detailed Complexity of MWBM}, page~\pageref{AppendixB: Detailed Complexity of MWBM}), let $L_i$ consists of edges of remaining graph $G$ (left after \,$i-1$-th iteration), whose weights
reduce in $G_h^\Delta$ in the $i$-th iteration of \textsc{Wt-Mwbm}(\,). 
Assume that there be $p$ iterations for the Algorithm~\ref{Algorithm1},
$l_i=|L_i|$ for
$i=1,2,\ldots,p  \leq N$ and $h_i=H_{i1}-H_{i2}$ in the \,$i$-th iteration, where $H_{i1}$ and $H_{i2} \,(<H_{i1})$ are the first two distinct maximum edge weights of the remaining graph $G$ after \,$i-1$-th iteration. 
From the detailed complexity analysis we have,  
\[
l_1h_1+l_2h_2+ \cdots + l_ph_p=W \quad\text{and}\quad 
l_1+l_2+ \cdots + l_p=W'.
\]

Observe that for the new graph $\widehat{G}$, the number of iterations in Algorithm \textsc{Wt-Mwbm}($\widehat{G}$) still remains $p$ and in the computation of \textsc{Wt-Mwbm}($\widehat{G}$), $L_i$ consists of $l_i$ number of edges  of the remaining graph $\widehat{G}$ (after \,$i-1$-th iteration), whose weights
reduce in $\widehat{G}_h^\Delta$ in $i$-th iteration of \textsc{Wt-Mwbm}(\,). 
In this case, if $h_i'=H'_{i1}-H'_{i2}$ in the $i$-th iteration, where $H'_{i1}$ and $H'_{i2} \;(<H'_{i1})$ are the first two distinct maximum edge weights of the remaining graph $\widehat{G}$ after \,$i-1$-th iteration, respectively, then
$$h_i'=H'_{i1}-H'_{i2}=\alpha*H_{i1} - \alpha*H_{i2}=
\alpha h_i \quad\text{where}\quad i=1,2,\ldots,p, $$ 
$$
l_1h'_1+l_2h'_2+ \cdots + l_ph'_p
= l_1 \alpha h_1+l_2 \alpha h_2+ \cdots + l_p \alpha h_p
=\alpha W, $$
and
$$ 
l_1+l_2+ \cdots + l_p=W'.
$$
Therefore, the modified Algorithm~\ref{Algorithm1} will take $O(\sqrt{|V|}W'/k(|V|,W'/{N}))$ time to compute the weight of a
$\textit{mwm}(\widehat{G})$ by using the algorithm by Feder and Motwani~\cite{feder95} as a subroutine; whereas time required for the Kao et. al.'s Algorithm~\ref{Algorithm0_Kao} 
is  
$O(\sqrt{|V|}\alpha W/k(|V|,\alpha W/N))$ time. 
\end{proof}

That is, multiplication by an integer constant to all the weight of edges of a weighted bipartite graph  $G$ does not affect the time complexity of the modified decomposition algorithm for computing the weight of a MWBM of the bipartite graph. The following remark talks about a conditional scaling down property of the algorithm for  the graph $G$.

\begin{remark}[Multiplicative Scaling Down Property]
Let we scale down each edge weights of 
$G$ by multiplying a factor of \,$\frac{1}{\alpha}$\, and get a new graph $\widehat{G}$, where $\alpha$ is the Greatest Common Divisor (GCD) of the positive edge weights of $G$. Then the time complexity of the Algorithm~\ref{Algorithm1} for computing a MWBM of both the graphs $G$ and $\widehat{G}$ remains same.
\end{remark}
%
%

Though during the complexity calculation of Algorithm~\ref{Algorithm1} we have stated a bound for $W'$ as: $|E| \leq W' \leq W$, but the following proposition gives a more better bound of the parameter $W'$.
\begin{proposition}[GCD Property]
\label{GCD_Property}
Let $G = (V, E, \textit{Wt})$ be an undirected, weighted bipartite graph and $E=\{e_1,e_2,\ldots,e_{|E|}\}$ be the set of positive weight edges with weights $\textit{Wt}(e_i)=w_i >0$ for $1 \leq i \leq |E|$. 
Further, let the GCD of the positive edge weights of $G$ is denoted by $ \textit{GCD}(w_1,w_2,\ldots,w_{|E|})$,
 then
$$|E| \leq W' \leq \frac{W}{ \textit{GCD}(w_1,w_2,\ldots,w_{|E|})} \leq W.$$
\end{proposition} 
\begin{proof}
Without going into more detailed and repeated writing, as mentioned 
in the previous Proposition~\ref{Lm:M_Incremental Wt}, we have:  
\begin{align*}
& l_1h_1+l_2h_2+ \cdots + l_ph_p=W, \text{~where~} h_i=H_{i1}-H_{i2} 
\quad\text{and}\quad \\
& l_1+l_2+ \cdots + l_p=W'.
\end{align*}

Let $g=$ GCD$(w_1,w_2,\ldots,w_{|E|})$. 
Observer that, in any iteration $i$ (where $i=1,2,\ldots,p$) both $H_{i1}$ and $H_{i2}$ are divisible by $g$. Hence, according to the definition of $h_i$s,  each $h_i=H_{i1}-H_{i2} $ is also divisible by the factor $g$.
\begin{align*}
\therefore~~W' 
& =    l_1+l_2+ \cdots + l_p  \\
& \leq l_1(h_1/g)+l_2(h_2/g)+ \cdots + l_p(h_p/g) \\
& \leq (l_1h_1+l_2h_2+ \cdots + l_ph_p)/g = \frac{W}{g} \leq W.
\end{align*}
This completes the proof.
\end{proof}


\subsection{Complexity Analysis by Considering Other Base Algorithms}
We also analyze the complexity of the Algorithm \ref{Algorithm1} by
considering the Hopcroft-Karp algorithm~\cite{hopcroft73} and
Alt-Blum-Mehlhorn-Paul algorithm~\cite{alt91} as base algorithms.
\begin{description}
\item[With Respect to the Hopcroft-Karp Algorithm:]
Hopcroft-Karp algorithm~\cite{hopcroft73} presents the best known worst-case performance
for getting a maximum matching in a bipartite graph with runtime of
$O(\sqrt{|V|}|E|)$. Hence the recurrence relation for running time of
the Algorithm~\ref{Algorithm1} with respect to Hopcroft-Karp algorithm is
$$
\begin{array}{ll}
& T(|V|,W',{N})= O(\sqrt{|V|}l_1) + T(|V|,W'',N'') \\[4pt] 
 \mbox{and}\quad &T(|V|,0,0)=0
\end{array}
$$
$$
\begin{array}{ll}
\therefore ~T(|V|,W',{N}) & = O(\sqrt{|V|}l_1) + O(\sqrt{|V|}l_2) + \cdots + O(\sqrt{|V|}l_p)\\[4pt]
 & = O\left(\sqrt{|V|}\sum\limits_{i=1}^p l_i\right) 
   = O(\sqrt{|V|}W').\\
\end{array}
$$

\item[With Respect to the Alt-Blum-Mehlhorn-Paul Algorithm:]
A bit better algorithm for dense bipartite graph is
Alt-Blum-Mehlhorn-Paul algorithm~\cite{alt91} which is 
$(\log |V|)^{1/2}$-factor faster
than Hopcroft-Karp algorithm for maximum
bipartite matching. Hence the time complexity, with respect to
Alt-Blum-Mehlhorn-Paul algorithm as a base algorithm, is
$O((|V|/\log |V|)^{1/2}W')$ and it is $(\log |V|)^{1/2}$-factor faster than the above case.

\end{description}

\section{Finding a Maximum Weight Matching} \label{Find_MWC}
The Algorithm~\ref{Algorithm1} computes
only the weight of a $\textit{mwm}(G)$ of a given graph $G$. To find the edges of a $\textit{mwm}(G)$, we 
give a revised
algorithm for constructing a Minimum Weight Cover (MWC) of $G$ which is a dual of maximum
weight matching.
As mentioned before, a \emph{cover} of $G$ is a function $C \cln V_1 \cup V_2 \rightarrow \nint$ such that
$C(v_1)+C(v_2) \geq \textit{Wt}(v_1,v_2) ~~ \forall v_1 \in V_1 \mbox{ and } v_2\in
V_2$. Let $\textit{Wt}(C)= \sum_{x\in V_1 \cup V_2} C(x)$. We say $C$ is \emph{minimum
weight cover} if $\textit{Wt}(C)$ is minimum. Let $C$ be a MWC of a graph $G$.

\begin{lemma}[\!\!\cite{kao02}] \label{Lemma_MWC}
Let $C_h^\Delta$ be any minimum weight cover of $G_h^\Delta$. If $C$ is
a function on $V(G)$ such that for every $u \in V(G)$, $C(u)=C_h(u) +
C_h^\Delta(u) $, then $C$ is minimum weight cover of $G$.
\end{lemma}

Using this lemma we design an $O(\sqrt{|V|}W'/k(|V|,W'/{N}))$-time
revised algorithm to compute a MWC of $G$. The correctness of this
algorithm is clear from the Lemma~\ref{Lemma_MWC} and the time
complexity analysis is similar to that given in the previous section.

\begin{algorithm}[H]
\caption{Calculate a MWC $C$ of $G$.}
\label{Algorithm2}
\begin{description}
\item[\it Input:] A weighted, undirected, complete bipartite graph $G$ with 
positive  integer weights on the edges.
\item[\it Output:] A minimum weight cover $C$ of $G$.
\end{description}

\begin{tabbing}
{\textsc{Mwc}($G$)}\\
123\=\+\kill
~1: Assume that initially $\textit{Wt}(\textit{mwm}(G))=0$.\\
~2: Find $h \leftarrow H_1-H_2$ from the current working graph $G$.\\
~3: Construct $G_h$ from $G$.\\
~4: Compute $\textit{mm}(G_h)$.\\
~5: Find minimum weight cover $C_h$ of $G_h$.\\
~6: Construct $G_h^\Delta$ from $G$ and $C_h$.\\
~7: \= {\bf if} $G_h^\Delta$ is empty (that is, $G_h^\Delta$ has no active edge)\\
~8:   \>~~~ {\bf then return}   $C_h$;\\
~9:   \>{\bf else}\\
10:      \>~~~ $C_h^\Delta$ $\leftarrow$  \textsc{Mwc}($G_h^\Delta$);\\
11:      \>~~~ {\bf return} $C$, where $C(u)=C_h(u)+C_h^\Delta(u)$ for all nodes $u$ in $G$.
\end{tabbing}
\end{algorithm}

Now as deduced by Kao et al.\ in~\cite{kao02}, finding a maximum weight matching by using the given vertex cover takes 
$O(\sqrt{|V|}|E|/k(|V|,|E|))$ time.
Since $|E| \leq W' \leq W,$ so altogether 
$O(\sqrt{|V|}W'/k(|V|,W'/{N}))$ time requires to find a MWBM of $G$.
\section{Experimental Evaluation} \label{mwbm:Experiments}
The Algorithm~\ref{Algorithm1} is efficient because of the modified decomposition Theorem~\ref{DecompositionTh_Modified}.
In order to understand the practical importance of the Algorithm~\ref{Algorithm1}, we report experimental evaluations of the same for the randomly generated weighted bipartite graphs. 

\subsection{Implementation and Experimental Environments}
We have implemented both Kao et al.'s algorithm~\cite{kao02}  and Algorithm~\ref{Algorithm1} in \verb|C++| and compiled them using \verb|g++| \textit{4.8.2-19ubuntu1} compiler.
All the experiments have been performed on a Desktop PC with an  
\textit{Intel}\textsuperscript{\textregistered} \textit{Xeon}\textsuperscript{\textregistered}(\textit{E5620 @ 2.40 GHz}) \textit{Processor}, \textit{32.00 GB RAM} and \textit{1200 GB Hard Disk}, running the \textit{Ubuntu 14.04.1 LTS} (\textit{Trusty Tahr}) \textit{64-bit Operating System}.

\subsection{Input Data Description and Its Randomness}
For a frame of fixed number of vertices in a partition of  the vertex set  and fixed weight of  bipartite graph $G$, we have generated the random weighted $G$ by 
assigning random (uniformly distributed) weight to the randomly (uniformly distributed) picked up edges of $G$. 
%
%
The outputs of these experiments for an input bipartite graph $G$ are:
\begin{enumerate}[(a)]
\item the number of iterations of \textsc{Wt-Mwbm}(\,) and Compute-\textsc{Mwm}(\,) in the Algorithms~\ref{Algorithm1} and~\ref{Algorithm0_Kao}, respectively, for the graph $G$ and
\item 
total time taken by the respective algorithms to compute the weight of a MWBM of $G$.  
\end{enumerate}
%


As mentioned in~\cite{hazay07},
%
the Approximate Parameterized String Matching (APSM) problem under Hamming distance error model is computationally equivalent to the MWBM problem in graph theory. 
The input data relation between the above problems are:
\begin{enumerate}[(a)]
\item length of the pattern is equal to  weight of the bipartite graph, and
\item alphabet size of the pattern is equal to number of vertices in a partition of the vertex set of the corresponding bipartite graph.  
\end{enumerate}

\subsection{Experimental Results}
We have tested the respective algorithms with large input data sets. The details are given below. In 
each of the graphs, the output of our Algorithm~\ref{Algorithm1} (denoted in short by ``Modified Algorithm'') corresponds to the  red colored unbroken line, whereas that of for the Algorithm~\ref{Algorithm0_Kao} (denoted in short by ``Kao et al.'s Algorithm'') corresponds to the red colored dotted line.

\begin{experiment}
\label{mwbm:Exp:1}
This experiment is done for a total of $250$ pseudo-randomly generated bipartite graphs, each of its weight is fixed to $1000$ unit where size of each of the partitions of the vertex set of bipartite graph varies from~$2$ to~$26$. 
%
\begin{table}[!htb]
\centering
{\scriptsize
\caption{Efficiency comparison between the Algorithms~\ref{Algorithm1} and~\ref{Algorithm0_Kao} for the 250 pseudo-randomly generated weighted bipartite graphs as considered in Experiment~\ref{mwbm:Exp:1}.}
\label{mwbm:Table:Exp1}
\begin{tabular}{|c|c|c|c|c|c|}
\hline
\multirow{2}{*}{\textbf{\begin{tabular}[c]{@{}c@{}}\# Vertices \\ in a Partition\end{tabular}}} & \multirow{2}{*}{\textbf{\begin{tabular}[c]{@{}c@{}}Weight\\ of Graph\end{tabular}}} & \multicolumn{2}{c|}{\textbf{Algorithm~\ref{Algorithm1}}} & \multicolumn{2}{c|}{\textbf{Algorithm~\ref{Algorithm0_Kao} (by Kao et al.)}} \\ \cline{3-6} 
 &  & \textbf{\# Iterations} & \textbf{Time (Sec.)} & \textbf{\# Iterations} & \textbf{Time (Sec.)} \\ \hline
2 & 1000 & 8.40 & 0.000050 & 422.60 & 0.002356 \\ \hline
3 & 1000 & 8.20 & 0.000084 & 397.60 & 0.003769 \\ \hline
4 & 1000 & 6.60 & 0.000104 & 402.20 & 0.005922 \\ \hline
5 & 1000 & 7.20 & 0.000146 & 444.40 & 0.009566 \\ \hline
6 & 1000 & 8.60 & 0.000214 & 408.60 & 0.010882 \\ \hline
7 & 1000 & 7.00 & 0.000251 & 418.60 & 0.014451 \\ \hline
8 & 1000 & 10.40 & 0.000405 & 455.20 & 0.019964 \\ \hline
9 & 1000 & 7.20 & 0.000376 & 366.40 & 0.018605 \\ \hline
10 & 1000 & 7.40 & 0.000411 & 400.20 & 0.024284 \\ \hline
11 & 1000 & 8.60 & 0.000525 & 391.20 & 0.025592 \\ \hline
12 & 1000 & 9.00 & 0.000670 & 391.20 & 0.031218 \\ \hline
13 & 1000 & 8.40 & 0.000722 & 391.20 & 0.035947 \\ \hline
14 & 1000 & 8.40 & 0.000811 & 391.20 & 0.039951 \\ \hline
15 & 1000 & 10.20 & 0.001089 & 391.20 & 0.044228 \\ \hline
16 & 1000 & 9.60 & 0.001156 & 391.20 & 0.048768 \\ \hline
17 & 1000 & 9.40 & 0.001217 & 391.20 & 0.053661 \\ \hline
18 & 1000 & 9.60 & 0.001411 & 391.20 & 0.058682 \\ \hline
19 & 1000 & 11.00 & 0.001680 & 391.20 & 0.064456 \\ \hline
20 & 1000 & 8.60 & 0.001492 & 391.20 & 0.070000 \\ \hline
21 & 1000 & 11.80 & 0.002226 & 391.20 & 0.077403 \\ \hline
22 & 1000 & 11.00 & 0.002284 & 391.20 & 0.083169 \\ \hline
23 & 1000 & 13.40 & 0.002924 & 391.20 & 0.089702 \\ \hline
24 & 1000 & 13.80 & 0.003254 & 391.20 & 0.097006 \\ \hline
25 & 1000 & 13.00 & 0.003327 & 391.20 & 0.103550 \\ \hline
26 & 1000 & 12.20 & 0.003342 & 391.20 & 0.110369 \\ \hline
\end{tabular}
}
\end{table}

Each numerical row of the Table~\ref{mwbm:Table:Exp1} is corresponding to 10  different random graphs, each of whose size of each partition of the vertex set and weight of the graphs are fixed.
Only for this experimental result, each row reports the average output  of $10$ different random graphs, each of whose number of vertices and weight are fixed.


For example, the numerical row corresponding to `$\#$\,Vertices in a partition' equal to $15$ reports the following. For the $10$ randomly generated different bipartite graphs, each of whose size of the vertex set is $30$ and weight is $1000$ unit. And on an average the number of iterations of \textsc{Wt-Mwbm}(\,) and Compute-\textsc{Mwm}(\,) in the Algorithms~\ref{Algorithm1} and~\ref{Algorithm0_Kao} are $10.20$ and $391.20$, respectively; whereas average time taken by the respective algorithms to compute the weight of a MWBM are $0.001089$ and $0.044228$ seconds.  

\begin{figure}[!htpb]
\centering
\includegraphics[angle=-90,width=.965\textwidth]{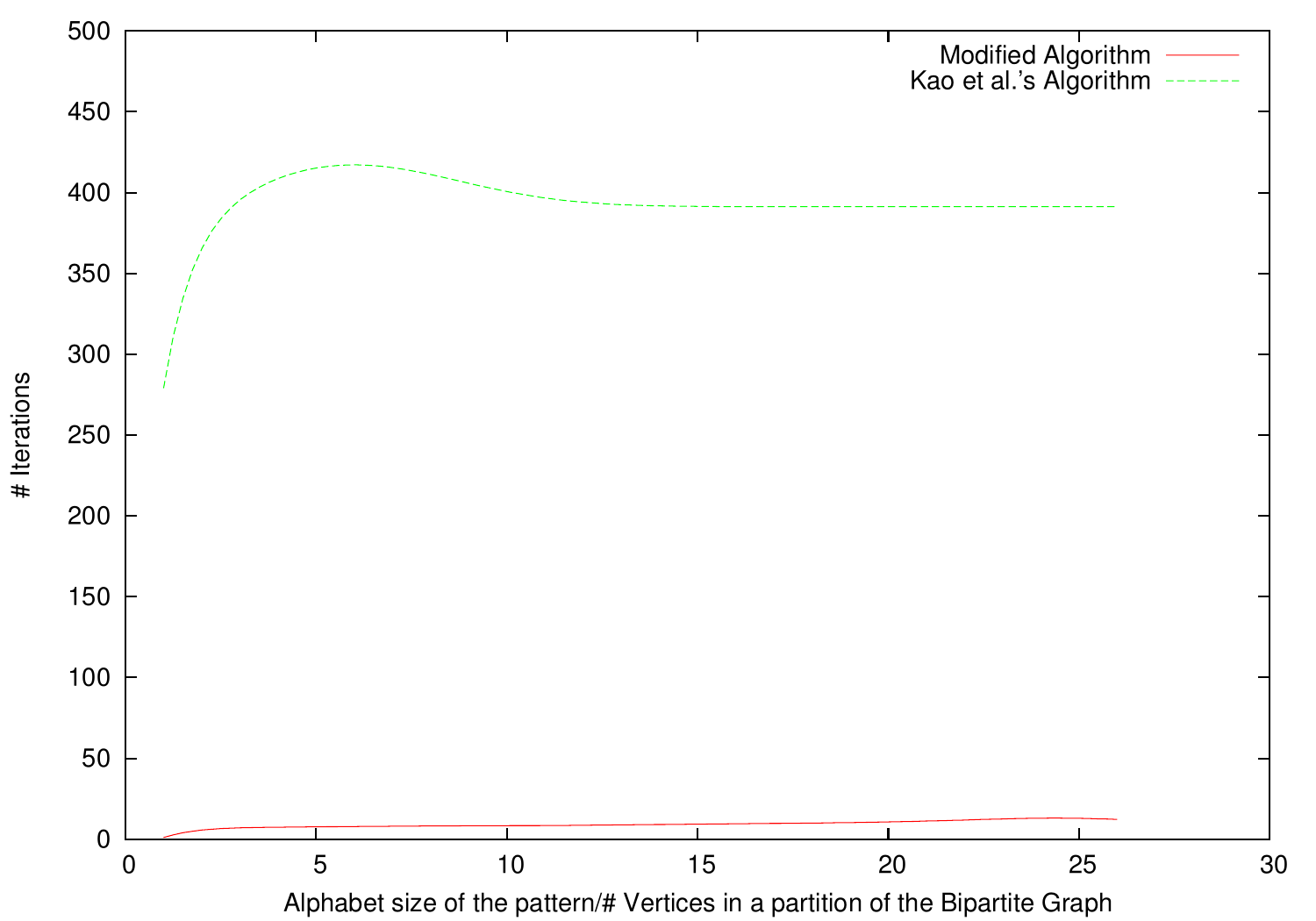}
\caption{Partition size vs.\ Iteration graph corresponding to the Experiment~\ref{mwbm:Exp:1}. Weight of each input graph is fixed to be 1000 unit.}
\label{mwbm:Fig:Exp1:Iteration}
\end{figure}
\begin{figure}[!htpb]
\centering
\includegraphics[angle=-90,width=.965\textwidth]{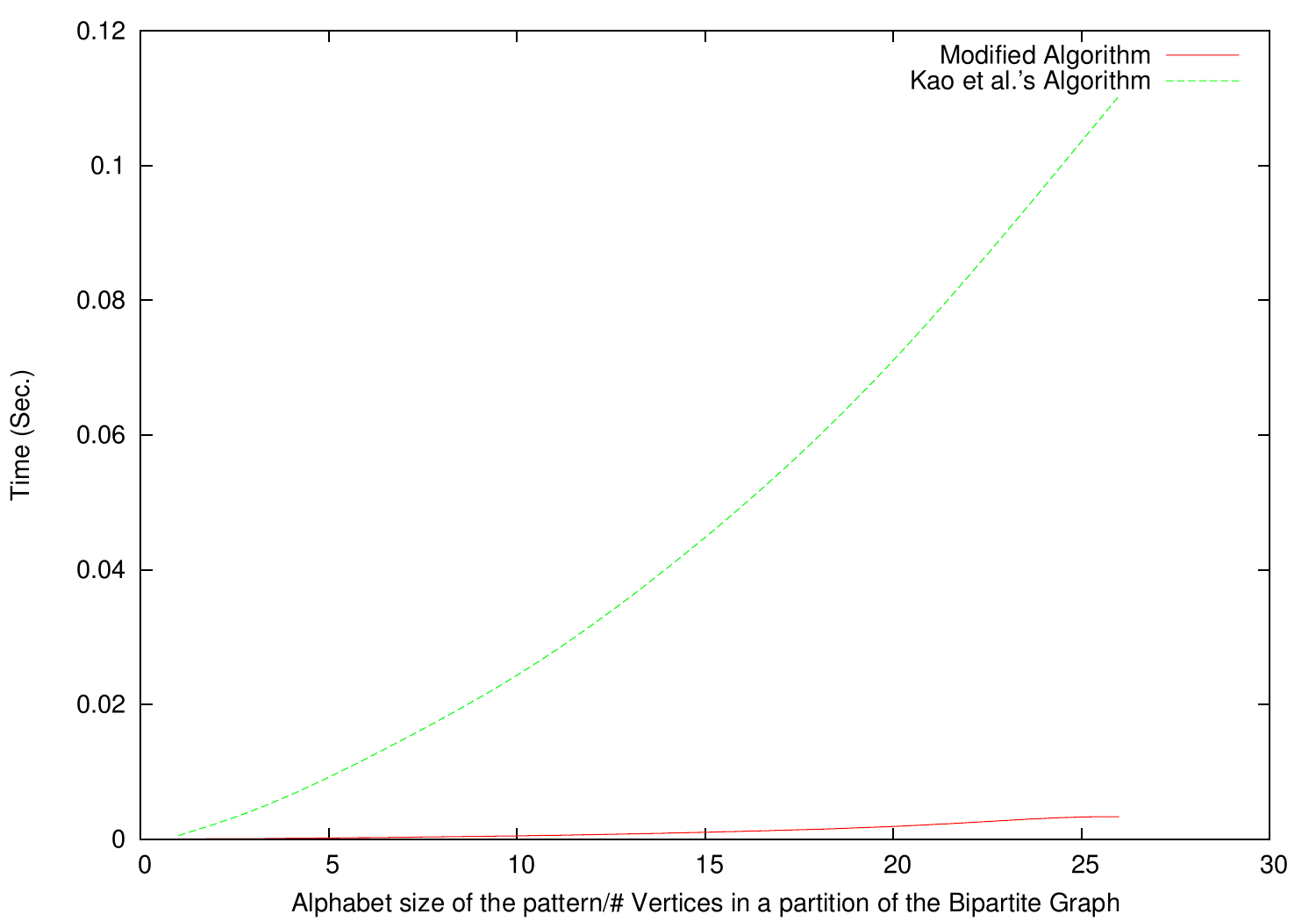}
\caption{Partition size vs.\ Time graph corresponding to the Experiment~\ref{mwbm:Exp:1}. Weight of each input graph is fixed to be 1000 unit.}
\label{mwbm:Fig:Exp1:Time}
\end{figure}

Figure~\ref{mwbm:Fig:Exp1:Iteration} shows the comparison on the number of iterations for the random graphs with different size partition of the vertices for the Algorithms~\ref{Algorithm1} and~\ref{Algorithm0_Kao}. Similarly, Figure~\ref{mwbm:Fig:Exp1:Time} gives the comparison time taken by the same algorithms for the random graphs with different size partition of the vertices. 
%
\end{experiment}
\newpage
The next two experiments are done over the  graphs corresponding to the randomly generated 
 strings over the DNA alphabet $\varSigma=\{A,C,G,T\}$ of different lengths.
%

\begin{experiment}
\label{mwbm:Exp:2}
In this experiment we have fixed the size of each partition of each graph to $4$ and randomly generated a total of $62$ bipartite graphs for $62$ different weights. See Table~\ref{mwbm:Table:Exp2} for more details. Unlike previous experiment, each row reports the iterations and time comparison of the Algorithms~\ref{Algorithm1} and~\ref{Algorithm0_Kao} on a randomly generated bipartite graph with fixed size vertex and weight. 
Figures~\ref{mwbm:Fig:Exp2:Iteration} and~\ref{mwbm:Fig:Exp2:Time} describe the pictorial representation of the Table~\ref{mwbm:Table:Exp2}.
%
\begin{table}[htpb]
\centering
{\scriptsize
\caption
[Experimental result for the 62 pseudo-randomly generated weighted bipartite graphs as considered in Experiment~\ref{mwbm:Exp:2}]
{Experimental result for the 62 pseudo-randomly generated weighted bipartite graphs as considered in Experiment~\ref{mwbm:Exp:2}. The number of vertices in each partition of the vertex set of each of the graphs is fixed to be~4, but weight of the graph varies.}
\label{mwbm:Table:Exp2}
\begin{tabular}{|c|c|c|c|c|c|}
\hline
\multirow{2}{*}{\textbf{\begin{tabular}[c]{@{}c@{}}\# Vertices \\in a Partition\end{tabular}}} & \multirow{2}{*}{\textbf{\begin{tabular}[c]{@{}c@{}}Weight\\ of Graph\end{tabular}}} & \multicolumn{2}{c|}{\textbf{Algorithm~\ref{Algorithm1}}} & \multicolumn{2}{c|}{\textbf{Algorithm~\ref{Algorithm0_Kao} (by Kao et al.)}} \\ \cline{3-6} 
 &  & \textbf{\# Iterations} & \textbf{Time (Sec.)} & \textbf{\# Iterations} & \textbf{Time (Sec.)} \\ \hline
4 & 10   & 3.00  & 0.000121 & 5.00    & 0.000152 \\ \hline
4 & 50   & 4.00  & 0.000109 & 13.00   & 0.000229 \\ \hline
4 & 100  & 9.00  & 0.000260 & 38.00   & 0.000965 \\ \hline
4 & 150  & 6.00  & 0.000165 & 52.00   & 0.001257 \\ \hline
4 & 200  & 5.00  & 0.000141 & 81.00   & 0.001641 \\ \hline
4 & 250  & 5.00  & 0.000142 & 109.00  & 0.002703 \\ \hline
4 & 300  & 18.00 & 0.000419 & 133.00  & 0.003255 \\ \hline
4 & 350  & 5.00  & 0.000144 & 116.00  & 0.002648 \\ \hline
4 & 400  & 6.00  & 0.000192 & 139.00  & 0.003666 \\ \hline
4 & 450  & 4.00  & 0.000122 & 192.00  & 0.004288 \\ \hline
4 & 500  & 31.00 & 0.000745 & 189.00  & 0.004863 \\ \hline
4 & 550  & 6.00  & 0.000165 & 203.00  & 0.004828 \\ \hline
4 & 600  & 6.00  & 0.000159 & 277.00  & 0.006564 \\ \hline
4 & 650  & 6.00  & 0.000166 & 298.00  & 0.006679 \\ \hline
4 & 700  & 9.00  & 0.000195 & 186.00  & 0.003468 \\ \hline
4 & 750  & 8.00  & 0.000222 & 268.00  & 0.007218 \\ \hline
4 & 800  & 7.00  & 0.000188 & 243.00  & 0.005828 \\ \hline
4 & 850  & 8.00  & 0.000193 & 390.00  & 0.009581 \\ \hline
4 & 900  & 8.00  & 0.000205 & 380.00  & 0.010258 \\ \hline
4 & 950  & 11.00 & 0.000297 & 395.00  & 0.011187 \\ \hline
4 & 1000 & 11.00 & 0.000248 & 368.00  & 0.009515 \\ \hline
4 & 1050 & 12.00 & 0.000230 & 466.00  & 0.012499 \\ \hline
4 & 1100 & 5.00  & 0.000135 & 535.00  & 0.012997 \\ \hline
4 & 1150 & 6.00  & 0.000126 & 405.00  & 0.008078 \\ \hline
4 & 1200 & 8.00  & 0.000205 & 397.00  & 0.008670 \\ \hline
4 & 1250 & 6.00  & 0.000168 & 602.00  & 0.014516 \\ \hline
4 & 1300 & 6.00  & 0.000168 & 648.00  & 0.016127 \\ \hline
4 & 1350 & 13.00 & 0.000283 & 553.00  & 0.015214 \\ \hline
4 & 1400 & 6.00  & 0.000158 & 639.00  & 0.016306 \\ \hline
4 & 1450 & 10.00 & 0.000242 & 426.00  & 0.009439 \\ \hline
4 & 1500 & 94.00 & 0.002554 & 456.00  & 0.011925 \\ \hline
4 & 1550 & 9.00  & 0.000267 & 591.00  & 0.017001 \\ \hline
4 & 1600 & 6.00  & 0.000162 & 775.00  & 0.016471 \\ \hline
4 & 1650 & 7.00  & 0.000189 & 676.00  & 0.015674 \\ \hline
4 & 1700 & 6.00  & 0.000162 & 665.00  & 0.013834 \\ \hline
4 & 1750 & 6.00  & 0.000162 & 860.00  & 0.021617 \\ \hline
4 & 1800 & 6.00  & 0.000163 & 701.00  & 0.017824 \\ \hline
4 & 1850 & 6.00  & 0.000150 & 777.00  & 0.016954 \\ \hline
4 & 1900 & 6.00  & 0.000142 & 827.00  & 0.019956 \\ \hline
4 & 1950 & 8.00  & 0.000194 & 788.00  & 0.018098 \\ \hline
4 & 2000 & 6.00  & 0.000168 & 690.00  & 0.016825 \\ \hline
4 & 2050 & 6.00  & 0.000175 & 584.00  & 0.011463 \\ \hline
4 & 2100 & 37.00 & 0.000864 & 727.00  & 0.016040 \\ \hline
4 & 2150 & 5.00  & 0.000139 & 585.00  & 0.009586 \\ \hline
4 & 2200 & 8.00  & 0.000186 & 747.00  & 0.016389 \\ \hline
4 & 2250 & 6.00  & 0.000175 & 897.00  & 0.022177 \\ \hline
4 & 2300 & 10.00 & 0.000252 & 1091.00 & 0.028346 \\ \hline
4 & 2350 & 6.00  & 0.000164 & 740.00  & 0.014673 \\ \hline
4 & 2400 & 8.00  & 0.000185 & 954.00  & 0.022879 \\ \hline
4 & 2450 & 6.00  & 0.000168 & 998.00  & 0.023334 \\ \hline
4 & 2500 & 18.00 & 0.000439 & 735.00  & 0.017273 \\ \hline
4 & 2550 & 5.00  & 0.000159 & 1086.00 & 0.023379 \\ \hline
4 & 2600 & 7.00  & 0.000200 & 1035.00 & 0.027333 \\ \hline
4 & 2650 & 6.00  & 0.000155 & 1313.00 & 0.032712 \\ \hline
4 & 2700 & 4.00  & 0.000133 & 1348.00 & 0.031344 \\ \hline
4 & 2750 & 10.00 & 0.000254 & 922.00  & 0.020268 \\ \hline
4 & 2800 & 9.00  & 0.000271 & 1030.00 & 0.028735 \\ \hline
4 & 2850 & 15.00 & 0.000361 & 1206.00 & 0.029624 \\ \hline
4 & 2900 & 42.00 & 0.000958 & 1144.00 & 0.026569 \\ \hline
4 & 2950 & 5.00  & 0.000153 & 1266.00 & 0.030598 \\ \hline
4 & 3000 & 6.00  & 0.000180 & 1249.00 & 0.033715 \\ \hline
4 & 3050 & 8.00  & 0.000216 & 1117.00 & 0.027538 \\ \hline
\end{tabular}
}
\end{table}
\begin{figure}[htpb]
\centering
\includegraphics[angle=-90,width=.97\textwidth]{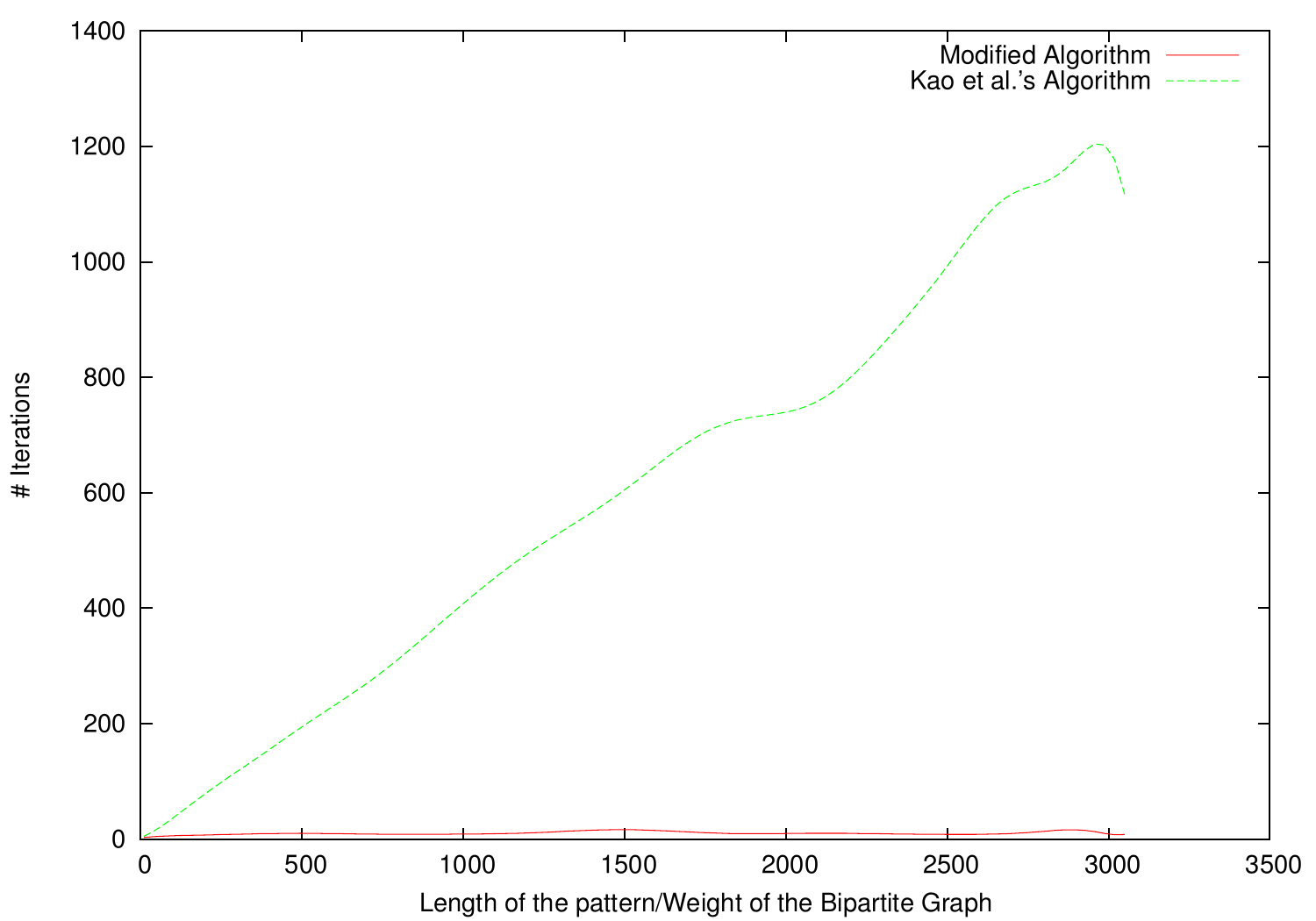}
\caption{Weight vs.\ Iteration graph corresponding to the Experiment~\ref{mwbm:Exp:2}. The number of vertices in each  partition of the vertex set is fixed to be  4.}
\label{mwbm:Fig:Exp2:Iteration}
\end{figure}
%
\begin{figure}[htpb]
\centering
\includegraphics[angle=-90,width=.97\textwidth]{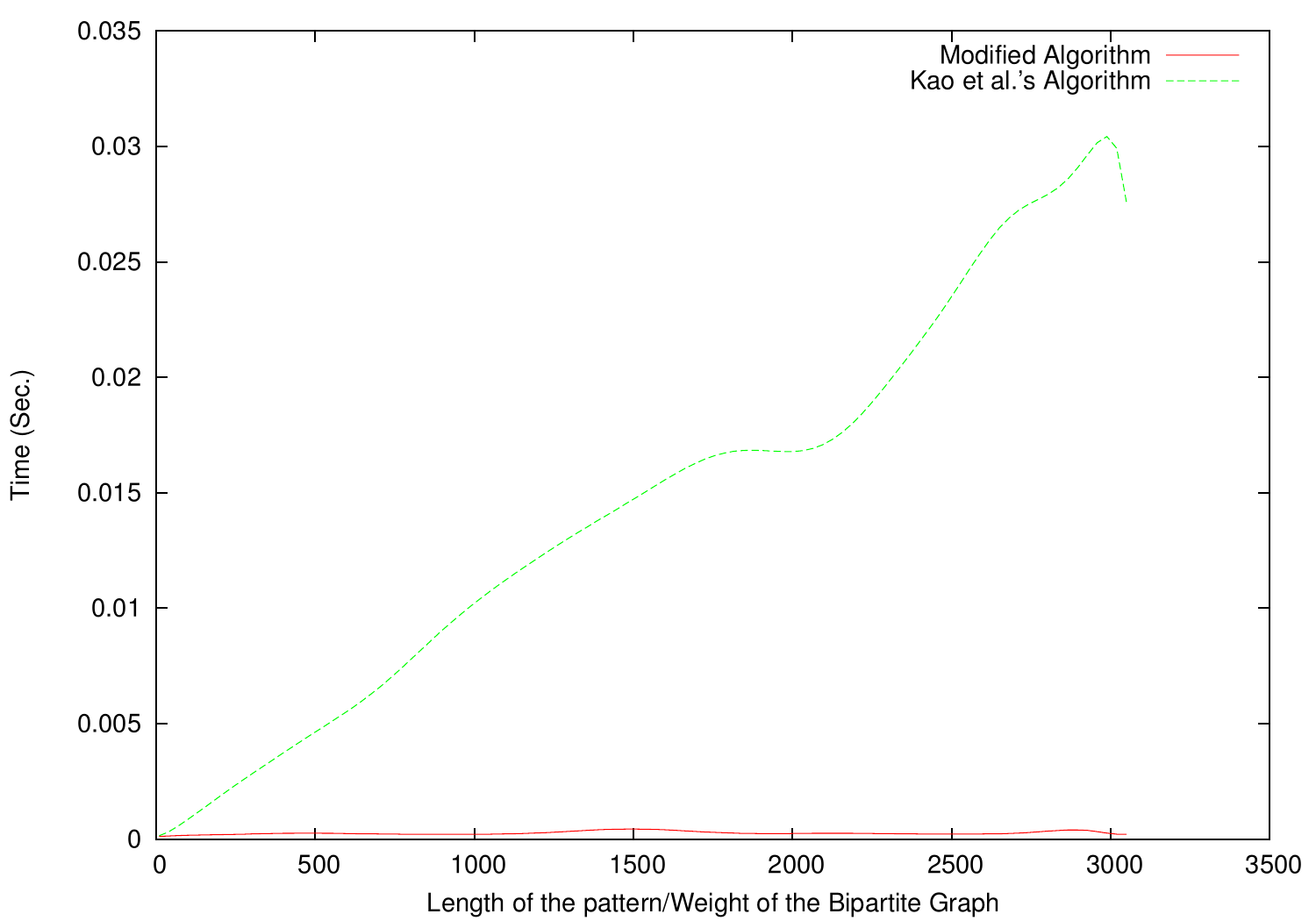}%
\caption{Weight vs.\ Time graph corresponding to the Experiment~\ref{mwbm:Exp:2}. The number of vertices in each  partition of the vertex set is fixed to be  4.}
\label{mwbm:Fig:Exp2:Time}
\end{figure}
\end{experiment}

\begin{experiment}
\label{mwbm:Exp:3}
In the final experiment also we have fixed the size of a partition of each graph to $4$ and but for a total of $71$ randomly generated bipartite graphs for $71$ different  and large weights. See Table~\ref{mwbm:Table:Exp3} for more details. 
%
\begin{table}[htpb]
\centering
{\tiny
\caption
[Experimental result for the 71 pseudo-randomly generated weighted bipartite graphs as considered in Experiment~\ref{mwbm:Exp:3}]
{Experimental result for the 71 pseudo-randomly generated 
bipartite graphs as considered in Experiment~\ref{mwbm:Exp:3}. 
Cardinality of each partition of the vertex set 
is fixed to be~4, but weight of the graph varies largely.}
\label{mwbm:Table:Exp3}
\begin{tabular}{|c|c|c|c|c|c|}
\hline
\multirow{2}{*}{\textbf{\begin{tabular}[c]{@{}c@{}}\# Vertices \\ in a Partition\end{tabular}}} & \multirow{2}{*}{\textbf{\begin{tabular}[c]{@{}c@{}}Weight\\ of Graph\end{tabular}}} & \multicolumn{2}{c|}{\textbf{Algorithm~\ref{Algorithm1}}} & \multicolumn{2}{c|}{\textbf{Algorithm~\ref{Algorithm0_Kao} (by Kao et al.)}} \\ \cline{3-6} 
 &  & \textbf{\# Iterations} & \textbf{Time (Sec.)} & \textbf{\# Iterations} & \textbf{Time (Sec.)} \\ \hline
4 & 1000   & 11.00 & 0.000241 & 368.00    & 0.009156 \\ \hline
4 & 10000  & 7.00  & 0.000196 & 4284.00   & 0.009156 \\ \hline
4 & 20000  & 8.00  & 0.000193 & 7722.00   & 0.009156 \\ \hline
4 & 30000  & 6.00  & 0.000161 & 14942.00  & 0.352380 \\ \hline
4 & 40000  & 12.00 & 0.000326 & 19786.00  & 0.589580 \\ \hline
4 & 50000  & 8.00  & 0.000213 & 22172.00  & 0.619935 \\ \hline
4 & 60000  & 8.00  & 0.000207 & 28763.00  & 0.806167 \\ \hline
4 & 70000  & 8.00  & 0.000215 & 22042.00  & 0.588712 \\ \hline
4 & 80000  & 9.00  & 0.000190 & 28054.00  & 0.662195 \\ \hline
4 & 90000  & 10.00 & 0.000240 & 21440.00  & 0.447115 \\ \hline
4 & 100000 & 36.00 & 0.000970 & 32975.00  & 0.865868 \\ \hline
4 & 110000 & 10.00 & 0.000222 & 53322.00  & 1.320364 \\ \hline
4 & 120000 & 8.00  & 0.000199 & 57741.00  & 1.466250 \\ \hline
4 & 130000 & 6.00  & 0.000156 & 54970.00  & 1.209947 \\ \hline
4 & 140000 & 9.00  & 0.000204 & 50849.00  & 1.102937 \\ \hline
4 & 150000 & 12.00 & 0.000221 & 65220.00  & 1.502660 \\ \hline
4 & 160000 & 9.00  & 0.000183 & 44405.00  & 0.868532 \\ \hline
4 & 170000 & 8.00  & 0.000188 & 71264.00  & 1.668321 \\ \hline
4 & 180000 & 9.00  & 0.000212 & 68565.00  & 1.391208 \\ \hline
4 & 190000 & 7.00  & 0.000182 & 58417.00  & 1.261169 \\ \hline
4 & 200000 & 8.00  & 0.000199 & 84882.00  & 2.177765 \\ \hline
4 & 210000 & 21.00 & 0.000411 & 55160.00  & 0.861068 \\ \hline
4 & 220000 & 7.00  & 0.000171 & 90913.00  & 1.931699 \\ \hline
4 & 230000 & 10.00 & 0.000247 & 87201.00  & 2.051072 \\ \hline
4 & 240000 & 15.00 & 0.000315 & 112402.00 & 2.794717 \\ \hline
4 & 250000 & 8.00  & 0.000180 & 102078.00 & 2.032557 \\ \hline
4 & 260000 & 17.00 & 0.000353 & 105322.00 & 2.745294 \\ \hline
4 & 270000 & 10.00 & 0.000251 & 88840.00  & 2.062840 \\ \hline
4 & 280000 & 15.00 & 0.000368 & 94300.00  & 2.191243 \\ \hline
4 & 290000 & 8.00  & 0.000191 & 79909.00  & 1.639328 \\ \hline
4 & 300000 & 18.00 & 0.000364 & 105128.00 & 2.443022 \\ \hline
4 & 310000 & 8.00  & 0.000188 & 120579.00 & 2.260459 \\ \hline
4 & 320000 & 8.00  & 0.000216 & 125597.00 & 3.134282 \\ \hline
4 & 330000 & 10.00 & 0.000236 & 151182.00 & 3.940926 \\ \hline
4 & 340000 & 9.00  & 0.000196 & 166492.00 & 3.507791 \\ \hline
4 & 350000 & 7.00  & 0.000181 & 151689.00 & 3.194340 \\ \hline
4 & 360000 & 7.00  & 0.000186 & 166928.00 & 3.850720 \\ \hline
4 & 370000 & 9.00  & 0.000212 & 167154.00 & 3.666068 \\ \hline
4 & 380000 & 10.00 & 0.000162 & 147368.00 & 2.885593 \\ \hline
4 & 390000 & 21.00 & 0.000388 & 137803.00 & 2.969532 \\ \hline
4 & 400000 & 9.00  & 0.000202 & 158026.00 & 3.690540 \\ \hline
4 & 410000 & 10.00 & 0.000256 & 153238.00 & 4.177567 \\ \hline
4 & 420000 & 15.00 & 0.000350 & 168440.00 & 4.543331 \\ \hline
4 & 430000 & 8.00  & 0.000174 & 195902.00 & 4.567199 \\ \hline
4 & 440000 & 7.00  & 0.000192 & 165922.00 & 4.078337 \\ \hline
4 & 450000 & 8.00  & 0.000199 & 183992.00 & 4.580142 \\ \hline
4 & 460000 & 8.00  & 0.000190 & 208746.00 & 4.302464 \\ \hline
4 & 470000 & 14.00 & 0.000234 & 229321.00 & 5.470295 \\ \hline
4 & 480000 & 8.00  & 0.000210 & 199475.00 & 5.379294 \\ \hline
4 & 490000 & 10.00 & 0.000266 & 239623.00 & 6.374744 \\ \hline
4 & 500000 & 11.00 & 0.000238 & 186026.00 & 4.691640 \\ \hline
4 & 510000 & 12.00 & 0.000298 & 201041.00 & 4.919859 \\ \hline
4 & 520000 & 9.00  & 0.000218 & 189501.00 & 4.317819 \\ \hline
4 & 530000 & 9.00  & 0.000205 & 151069.00 & 3.144941 \\ \hline
4 & 540000 & 8.00  & 0.000219 & 230280.00 & 5.777490 \\ \hline
4 & 550000 & 7.00  & 0.000164 & 254817.00 & 5.943550 \\ \hline
4 & 560000 & 20.00 & 0.000433 & 240510.00 & 6.015034 \\ \hline
4 & 570000 & 7.00  & 0.000178 & 260619.00 & 5.542112 \\ \hline
4 & 580000 & 7.00  & 0.000177 & 230100.00 & 4.817907 \\ \hline
4 & 590000 & 11.00 & 0.000239 & 240188.00 & 5.613709 \\ \hline
4 & 600000 & 8.00  & 0.000202 & 260924.00 & 5.703650 \\ \hline
4 & 610000 & 7.00  & 0.000179 & 261019.00 & 6.318097 \\ \hline
4 & 620000 & 13.00 & 0.000295 & 281111.00 & 5.626782 \\ \hline
4 & 630000 & 9.00  & 0.000227 & 276200.00 & 5.939048 \\ \hline
4 & 640000 & 10.00 & 0.000227 & 286193.00 & 7.329997 \\ \hline
4 & 650000 & 10.00 & 0.000213 & 255659.00 & 5.345079 \\ \hline
4 & 660000 & 6.00  & 0.000165 & 321451.00 & 8.061329 \\ \hline
4 & 670000 & 7.00  & 0.000170 & 243797.00 & 5.014862 \\ \hline
4 & 680000 & 8.00  & 0.000173 & 286228.00 & 6.904159 \\ \hline
4 & 690000 & 19.00 & 0.000396 & 306680.00 & 8.027268 \\ \hline
4 & 700000 & 4.00  & 0.000134 & 330990.00 & 7.337250 \\ \hline
\end{tabular}
}
\end{table}
\end{experiment}
\break
\section{Conclusions} 
\label{mwbm:Conclusion}
We have fine-tuned the existing decomposition theorem
originally proposed by Kao et al.\ in~\cite{kao02}, in
the context of maximum weight bipartite matching and applied it
 to design a revised version of the decomposition algorithm
to compute the weight of a maximum weight
bipartite matching in $\textit{O}(\sqrt{|V|}W'/k(|V|,W'/{N}))$ time by
employing an algorithm designed by Feder and Motwani \cite{feder95}, as
base algorithm. We have also analyzed the algorithm by using Hopcroft-Karp algorithm \cite{hopcroft73} and Alt-Blum-Mehlhorn-Paul algorithm \cite{alt91} as
base algorithms, respectively.

The algorithm performs well especially when the largest edge weight differs  by
more than one from the second largest edge weight in the current working graph
during an invocation of \textsc{Wt-Mwbm}(\,) in any iteration.
Further, we have given a scaling property of the algorithm and a bound of the parameter $W'$ as $|E| \leq W' \leq \frac{W}{ \textit{GCD}(w_1,w_2,\ldots,w_{|E|})} \leq W$, where $\textit{GCD}(w_1,w_2,\ldots,w_{|E|})$ denotes the GCD of the positive edges weights $\{w_1,w_2,\ldots,w_{|E|}\}$ of the weighted bipartite graph. 
%
The algorithm works well for
general $W,$ but is the best known for $W'=o(|E| \log(|V|N))$.
The experimental study shows that performance of the modified decomposition algorithm is satisfactory.

\subsubsection*{Acknowledgement:}
I am grateful to Dr. Kalpesh Kapoor for his helpful comments and suggestions. Also I thank Rahul
Kadyan, particularly for assisting me in the implementation of graph matching algorithms.




\section{Appendix: Detailed Complexity Analysis of Algorithm~\ref{Algorithm1}}
\label{AppendixB: Detailed Complexity of MWBM}
Here we give complexity analysis of the Algorithm
\ref{Algorithm1} in general. It is almost similar as done in the paper~\cite{kao02}. We assume that a maximum heap~\cite{cormen09} is used to store the distinct edge weights along with the associated edges of~$G$.

Let the running time of \textsc{Wt-Mwbm}($G$) be $T(|V|,W',N)$ excluding the
initialization. Let $L$ be the set of the heaviest weight edges in $G$.
So up to the Step 3, construction of $G_h$ requires $O(|L|\log |E|)$ time.
The Step 4 takes $O(\sqrt{|V|} |L|/k(|V|,|L|))$ time by using Feder and
Motwani's algorithm~\cite{feder95} to compute $\textit{mm}(G_h)$. In Step 5,
$C_h$ can be found in $O(|L|)$ time from this matching. Let $L_1$ be the
set of edges of $G$ adjacent to some node $u$ with $C_h(u)>0$, i.e.,
$L_1$ consist of edges of $G$ whose weights reduce in $G_h^\Delta$. Let
$l_1=|L_1|$. Step 6 updates every edges of $L_1$ in the heap in $O(l_1 \log |E|)$ time. Since $L \subseteq L_1$, Step 1 to 6 takes
$O(\sqrt{|V|}l_1/k(|V|,l_1))$ time altogether. Let $l_i=|L_i|$ for
$i=1,2,\dots,p \leq N$ and $h_i=H_1-H_2$ for $i$-th phase of 
the recursion, where $L_i$ consists of edges of remaining $G$ whose weights 
reduce in $G_h^\Delta$ on $i$-th iteration.
 Note that, 
 $$l_1h_1+l_2h_2+ \cdots + l_ph_p=W.$$ 
 
 Let $l_1+l_2+ \cdots + l_p=W'$. Observe that if
$h_i=1$ for all $i \in [1,p]$, then $W'=\sum_{i=1}^p l_i=W$. Step 7
uses at most $T(|V|,W'',N'')$ time, where $W''\;(<W')$ is the total weight of $G_h^\Delta$ and $N''\;(<N)$ is the
maximum edge weight of $G_h^\Delta$. Hence the recurrence
relation for running time is
\begin{align*}
& T(|V|,W',{N})= O(\sqrt{|V|}l_1/k(|V|,l_1)) + T(|V|,W'',N'') \quad\mbox{and}\\[5pt]
&T(|V|,0,0)=0
\end{align*}
\begin{align*}
 \therefore &~~T(|V|,W',{N}) \\[4pt]
	&~~~= O\left(\dfrac{\sqrt{|V|}l_1}{k(|V|,l_1)}\right) + O\left(\dfrac{\sqrt{|V|}l_2}{k(|V|,l_2)}\right) + \cdots + O\left(\dfrac{\sqrt{|V|}l_p}{k(|V|,l_p)}\right) \\[4pt]
	&~~~= O\left(\sqrt{|V|} \left(\dfrac{l_1}{k(|V|,l_1)}+\dfrac{l_2}{k(|V|,l_2)}+\cdots+\dfrac{l_p}{k(|V|,l_p)}\right)\right) \\[4pt]
	&~~~= O\left(\dfrac{\sqrt{|V|}}{\log |V|} \left(\log |V|^2 \sum \limits_{i=1}^p l_i - \sum \limits_{i=1}^p l_i \log l_i \right)\right)
\end{align*}

Let $f(x)=x \log{x}$. Note that it is a convex function, so by Jensen's inequality\footnote{\textbf{Jensen's Inequality}~\cite{hardy34}. If
$f(x)$ is a convex function on an interval $I$ and $\mu_1,
\mu_2,\ldots,\mu_n$ are positive weights such that $\sum_{i =
0}^{n} \mu_i=1$ then $$f\left(\sum_0^n \mu_i x_i\right) \leq
\sum_{i=0}^n  \mu_i f(x_i).$$},
$$\begin{array}{l}
\sum\limits_{i=1}^p l_i \log l_i = \sum\limits_{i=1}^p f(l_i) \geq p f\left(\dfrac{\sum\limits_{i=1}^p l_i}{p}\right)
 = p f\left(\dfrac{W'}{p}\right)\\[25pt]
 = p \dfrac{W'}{p} \log \dfrac{W'}{p} = W' \log \dfrac{W'}{p} = O\left(W' \log \dfrac{W'}{N}\right)
\end{array}.$$
%
This lead to the running time complexity as follows.
$$
\begin{array}{ll}
&T(|V|,W',N) 	\\[5pt]
&~~~= O\left(\dfrac{\sqrt{|V|}}{\log |V|} \left( W'\log {|V|^2} -  W' \log \dfrac{W'}{{N}} \right)\right) \\[13pt]
&~~~= O\left(\dfrac{\sqrt{|V|}W'}{\log |V|/\log \frac{|V|^2}{W'/{N}}}\right)\\[5pt]
&~~~= O(\sqrt{|V|}W'/k(|V|,W'/{N})).
\end{array}
$$
This is better than the $O(\sqrt{|V|}W/k(|V|,W/N))$ time as
mentioned in~\cite{kao02}. 

The parameter $W'$ is smaller than $W$ which is the total weight of $G$, 
 essentially when the heaviest edge weight differs by more than one unit from the
second heaviest edge weight in a current working graph 
during a decomposition in any iteration of the algorithm. 
In best case the algorithm
takes $O(\sqrt{|V|}|E|/k(|V|,|E|))$ time to compute a maximum weight matching  
and in worst case $O(\sqrt{|V|}W/k(|V|,W/N))$,
that is,  $|E| \leq W' \leq W$.
This time complexity bridges a gap between the best known time complexity for computing a 
Maximum Cardinality Matching (MCM) of unweighted bipartite graph and that of computing a MWBM of a weighted bipartite graph.

However, it is very difficult and challenging to get rid of $W$ or $N$ from the complexity.
This modified algorithm works well for general $W,$ but is  best known for $W'=o(|E| \log(|V|N))$. 



\end{document}